\documentclass[journal]{IEEEtran}

\usepackage{latexsym,amsmath,amssymb}
\def\F {{\mathbb{F}}}

\def\cX{{\mathcal X}}

\def\cL{{\mathcal L}}

\def\cP{{\mathcal P}}

\def\bc{{\bf c}}

\def\bo{{\bf 0}}

\def\bx{{\bf x}}

\def\beq{\begin{equation}}
\def\eeq{\end{equation}}

\def\c {{\mathbf{c}}}

\newtheorem{thm}{Theorem}[section]
\newtheorem{defn}[thm]{Definition}
\newtheorem{prop}[thm]{Proposition}
\newtheorem{lem}[thm]{Lemma}
\newtheorem{cor}[thm]{Corollary}
\numberwithin{equation}{section} \newtheorem{rem}[thm]{Remark}

\begin{document}

\onecolumn

\title{Erasure List-Decodable Codes from Random and Algebraic Geometry Codes}

\author{Yang Ding, Lingfei~Jin and Chaoping~Xing
\thanks{All authors are with Division of
Mathematical Sciences, School of Physical and Mathematical Sciences,
Nanyang Technological University, Singapore 637371, Republic of
Singapore (email: \{dingyang,lfjin,xingcp\}@ntu.edu.sg).}

\thanks{The work is partially supported by the Singapore   A*STAR   SERC   under   Research   Grant
1121720011.}}

\maketitle

\begin{abstract}
Erasure list decoding was introduced to correct a larger number of erasures with output of a list of possible candidates. In the present paper, we consider both random linear codes and algebraic geometry codes for list decoding erasure errors. The contributions of this paper are two-fold. Firstly,    we show that, for arbitrary $0<R<1$ and $\epsilon>0$ ($R$ and $\epsilon$ are independent), with high probability a random linear code is an erasure list decodable code with constant list size $2^{O(1/\epsilon)}$ that can correct a fraction $1-R-\epsilon$ of erasures, i.e., a random linear code achieves the information-theoretic optimal trade-off between information rate and fraction of erasure errors.   Secondly, we show that algebraic geometry codes are good erasure list-decodable codes. Precisely speaking, for any $0<R<1$ and $\epsilon>0$, a $q$-ary algebraic geometry code of rate $R$ from the Garcia-Stichtenoth tower can correct $1-R-\frac{1}{\sqrt{q}-1}+\frac{1}{q}-\epsilon$ fraction of erasure errors with list size $O(1/\epsilon)$. This improves the Johnson bound  applied to algebraic geometry codes. Furthermore,  list decoding of these algebraic geometry codes can be implemented in polynomial time.
\end{abstract}


\section{Introduction}
Erasure codes have received great attentions for their wide applications in recovering packet losses in the internet and  storage systems. In the model of erasure channel, errors are described as erasures, namely  the receivers are supposed to know the positions where the erasures occurred. Compared with other communication channels such as  adversarial noise channel,  erasure channel is much simpler. Thus, we can expect better parameters for erasure channel than  adversarial noise channel.  Instead of the unique  decoding, the model of list decoding for which a decoder allows to output a list of possible codewords  was independently introduced by Elias and Wonzencraft \cite{El, Wo}. The decoding is considered to be successful as long as the correct codeword is included in the list and the list size is not too big.

The problem of list decoding for classical adversarial noise channel has been extensively studied (see \cite{El,Gru01,Gru2, Gru3, Su,Wo, ZP}, for example). A fundamental problem in list decoding is the tradeoff among the information rate, decoding radius (i.e., fraction of errors that can be corrected) and the list size. In other words, if we fix one of these three parameters, then one is interested in optimal tradeoff between the remaining two parameters. For instance, if the list size is fixed to be constant or polynomial in the length of codes, the problem becomes a tradeoff between  information rate and decoding radius.
 \begin{defn}\label{1.1}($(\tau,L)$-erasure list decodability)   Let $\Sigma$ be a finite alphabet of size $q$, $L > 1$ an
integer, and  $\tau\in (0, 1)$. A code $C \subseteq \Sigma^n$
is said to be $(\tau, L)$-erasure list-decodable, if   for every $\mathbf{r}\in\F_q^{(1-\tau)n}$, and any subset $T\subseteq\{1,2,\cdots,n\}$ of size $(1-\tau)n$,  one has
 $$|\{\mathbf{c}\in C|\mathbf{c}_T=\mathbf{r}\}|\leq L,$$
 where $\c_T$ is the projection of $\c$ onto the coordinates indexed by $T$. In other words, given any received word with at most $\tau n$ erasures, there are at most $L$ codewords that are consistent  with the unerased portion of the received word.
 \end{defn}

\subsection*{Known results}

It is known that, for an erasure channel where the codeword symbols are randomly and independently erased with probability $\tau$, the capacity  is $1-\tau$ (see \cite{El1}). Although erasure list decoding has been considered previously (see \cite{Gu, GS, Gru2, Gru3}),  a lot of problems still remain unsolved. Let us summarize some of  previous results on  erasure list decoding below.
 \begin{itemize}
 \item[(i)] It was shown in \cite{Gu} that, for any small $\epsilon>0$ and  $\tau\in(0,1)$, a $(\tau,L)$-erasure list-decodable code of rate $1-\tau-\epsilon$ must satisfy $L\ge \Omega(\frac1{\epsilon})$; and on the other hand, there exists a  $(\tau,O(\exp(\frac1{\epsilon})))$-erasure list-decodable code of rate $1-\tau-\epsilon$.
 \item[(ii)]
In \cite[Proposition 10.1]{Gru01},  the Johnson bound for erasure decoding radius was derived. It says that, for any given $\epsilon>0$,  every $q$-ary code of relative distance $\delta<1-1/q$ is $(\delta+\frac{\delta}{q-1}-\epsilon, O(1/\epsilon))$-erasure list-decodable. This means that, with a constant list size,  erasure decoding radius is enlarged by approximaly $\frac{\delta}{q-1}$ compared with  unique erasure decoding whose decoding radius is only $\delta$. On the other hand, it was shown further in \cite[Proposition 10.2]{Gru01} that there exists a $q$-ary code of length $n$ and relative distance $\delta<1-1/q$ that is not $(\delta+\frac{\delta}{q-1}+\epsilon, 2^{\Omega(\epsilon^2\delta n)})$-erasure list-decodable for every small $\epsilon>0$. This implies that the best bound on erasure list decoding radius of a $q$-ary code of relative minimum distance $\delta$ is $\delta+\frac{\delta}{q-1}$.
 \item[(iii)] In \cite{Gu},
Guruswami showed that, for any small $\epsilon>0$, with high probability a  random linear code  of rate $R=\Omega(\epsilon/\log(1/\epsilon))$ is $(1-\sigma, O(1/\sigma))$-erasure list-decodable for every $\sigma$ satisfying $\epsilon\leq\sigma\leq1$. Furthermore, by the concatenation method  Guruswami showed in \cite{Gu} that, for any small $\epsilon>0$, one can  construct a family of concatenated (binary) $(1-\epsilon, O(1/\epsilon))$-erasure list-decodable codes of rate $\Omega(\epsilon^2/\log(1/\epsilon))$ in polynomial time.  A slightly better rate was obtained  for nonlinear codes over larger alphabet size in \cite{Gru3}.
\end{itemize}
\subsection*{Our results and comparison}
Our contributions of this paper are two-fold.
 \begin{itemize}
 \item[(i)] Firstly, we show that, for arbitrary $0<R<1$ and $\epsilon>0$ ($R$ and $\epsilon$ are independent), with high probability a random linear code is  $(1-R-\epsilon, 2^{O(1/\epsilon)})$-erasure list-decodable, i.e., a random linear code achieves the information-theoretic optimal tradeoff between information rate and fraction of erasure errors that can be corrected.
 While Theorem 2 in  \cite{Gu} which was derived from \cite{He} only shows existence of  $(1-R-\epsilon, 2^{O(1/\epsilon)})$-erasure list-decodable codes for arbitrary $0<R<1$ and $\epsilon>0$.
 \item[(ii)]
  Secondly, we show that algebraic geometry codes are good erasure list-decodable codes. Precisely speaking, for any $0<\tau<1$ and $\epsilon>0$, a $q$-ary algebraic geometry code from the Garcia-Stichtenoth tower has rate at least $1-\tau-\frac{1}{\sqrt{q}-1}+\frac{1}{q}-\epsilon$ and  is $(\tau, O(1/\epsilon))$-erasure list-decodable. Furthermore,  list decoding  of these algebraic geometry codes can be implemented in polynomial time. On the other hand, if we apply the Johnson bound given in \cite[Proposition 10.1]{Gru01} to general algebraic geometry codes, we can only claim that a $q$-ary
  algebraic geometry code from the Garcia-Stichtenoth tower has rate $1-\tau-\frac{1}{\sqrt{q}-1}+\frac{\tau}{q}-\epsilon$ and  is $(\tau, O(1/\epsilon))$-erasure list-decodable. This rate is always smaller than our rate for any $\tau\in(0,1)$. This implies that the Johnson bound could be improved for some special class of codes although it is optimal in general.
\end{itemize}

\subsection*{Open problems}
 For adversarial error channel,  it has been shown that, given decoding radius $0<\tau<1$, the optimal rate for list decoding  is $R=1-H_q(\tau)$, where
$H_q(x)=x\log_q(q-1)-x\log_q x-(1-x)\log_q (1-x)$ is the $q$-ary entropy function. More precisely speaking, for any small $\epsilon>0$ and $\tau$ with $0<\tau<1-1/q$,  with high probability a random code is $(1-H_q(\tau)-\epsilon, O(\frac1{\epsilon}))$-list decodable. Furthermore, every $q$-ary $(1-H_q(\tau)-\epsilon,L)$-list-decodable code has list size at least $\Omega(\log 1/\epsilon)$. It is still an open problem to determine if there exists a $q$-ary $(1-H_q(\tau)-\epsilon,L)$-list decodable code with list size $L$ smaller than $O(\frac1{\epsilon})$. Under the situation of erasure list decoding, the optimal rate $R$ that one can achieve is $R=1-\tau$. If we denote $L_{\tau,q}(\epsilon)$ to be the smallest integer $L$ for which there are
$q$-ary $(\tau, L)$-erasure list-decodable codes of rate at least $1 - \tau - \epsilon$ for infinitely many lengths $n$, then it follows from our result and \cite{Gu} that $\Omega(\frac1{\epsilon})\le L_{\tau,q}(\epsilon)\le 2^{O(1/\epsilon)}$. Now the first open problem is
\begin{quote}
{\it Open Problem 1:} Determine $L_{\tau,q}(\epsilon)$.
\end{quote}

In the literature, there are not many results on constructive bounds on erasure list decoding  except for sufficiently large $q$ or small rate \cite{Gu,GS}. The second open problem would be
\begin{quote}
{\it Open Problem 2:} Narrow the rate gap between $1-\tau-\frac{1}{\sqrt{q}-1}+\frac{1}{q}$ and $1-\tau$ by constructing erasure list-decodable codes explicitly, i.e., construct a $q$-ary $(\tau,L)$-erasure list-decodable codes of rate $R>1-\tau-\frac{1}{\sqrt{q}-1}+\frac{1}{q}$ such that the list size $L$ is either a constant or a polynomail in length.
\end{quote}

\subsection*{Organization}
The paper is organized as follows. In Section 2, we introduce some necessary natation and definitions and known results as well. Section 3 is devoted to random codes. In the last section, we show that algebraic geometry codes are good erasure list-decodable codes.

\section{Preliminaries}
In this paper, we only focus on  linear codes. Recall that a $q$-ary $[n,k]_q$ linear  code is an $\F_q$-linear subspace  of $\F_q^n$ with dimension $k$, where $\F_q$ is a finite field with $q$ elements and $q$ is a prime power. $n$ is called the length of the code and $k$ is the dimension of the code.  The information rate of the code $C$ is defined as $R=k/n$ which represents the efficiency of the code. Another important parameter of the code is the distance which represents the error correcting capability. The distance of a linear code $C$ is defined to be the minimum Hamming weight of  nonzero codewords of $C$, denoted by $d=d(C)$. The relative distance $\delta=\delta(C)$ is defined to be the quotient $d/n$.

 From Definition \ref{1.1}, one knows that, in a $(\tau,L)$-erasure list decodable code $C$ of length $n$, for every $\mathbf{r}\in\F_q^{(1-\tau)n}$ and $T\subseteq\{1,2,\dots,n\}$ with $|T|=(1-\tau)n$ the number of the codewords in the output list that are consistent with $\mathbf{r}$ at the coordinates indexed by $T$  is at most $L$.
Thus, if $C$ is linear,  it is equivalent  to saying  that  the number of the codewords  that are $\bo$ at the coordinates indexed by $T$ is at most $L$, i.e., $|\{\mathbf{c}\in C|\mathbf{c}_T=\bo\}|\leq L$. Hence, an $[n,k,d]_q$-linear code is $((d-1)/n,1)$-erasure list-decodable,  but  not $(d/n,1)$-erasure list-decodable.

 \begin{defn}\label{2.2}(Erasure list decoding radius (ELDR))
 \begin{itemize}
 \item[(i)] For an integer $L\geq 1$ and a linear code $C$ of length $n$, we denote
 \[{\rm Rad}_L(C):=\max\{s\in\mathbb{Z}_{>0}:\; \mbox{$C$ is $(s/n,L)$-erasure list-decodable}\}.\]
 \item[(ii)] For an infinite family $\mathcal{C}=\{C_i\}_{i\geq 1}$ of $q$-ary linear codes with length tending to $\infty$ and an integer $L\geq 1$, we denote $${\rm ELDR}_L(\mathcal{C}):=\liminf_{i}\left\{\frac{{\rm Rad}_L(C_i)}{n_i}\right\},$$ where $n_i$ is the length of $C_i$.
  \end{itemize}\end{defn}

  \begin{defn}\label{2.3}For an integer $L\geq 1$ and $0\leq \tau\leq 1$, the maximum rate for linear $(\tau, L)$-erasure list-decodable code families is defined to be $$R_L(\tau):=\sup_{\mathcal{C}:\; {\rm ELDR}_L(\mathcal{C})\geq\tau}R(\mathcal{C}).$$
  \end{defn}

 The notation of erasure list decoding  for linear codes actually had already been studied in the form of generalized Hamming weight, see \cite{wei}.
However, the explicit relationship  between erasure list decoding and generalized Hamming weight had not been made clear until the work in \cite{Gu}. The concept of generalized Hamming weight was initially introduced in \cite{wei} and later  received great attention  due to applications in cryptography, design of codes, $t$-resilient functions and so on \cite{As}.

 \begin{defn}\label{2.4}(Generalized Hamming Weight) The $r$-th generalized Hamming weight of a code $C$, denoted by $d_r(C)$, is defined to be the size of the smallest support of an $r$-dimensional subcode of  $C$, i.e.,
 \[d_r(C)=\min \{{\rm |Supp}(D)|:\;\mbox{$D$ is a subspace of $C$ of dimension $r$}\},\]
  where ${\rm Supp}(D)=\{i: \; \exists (c_1,\dots,c_n)\in D,  c_i\neq 0\}.$
 \end{defn}

Note that $d_1(C)$ is exactly the minimum distance $d$ of $C$. The characterization of erasure list decodability through generalized Hamming weight is given  below.

  \begin{lem} (see \cite{Gu})\label{2.5}
A linear code  $C$ of length $n$ is $(s/n,L)$-erasure list-decodable if and only if $d_r(C)>s$, where $r=\lfloor\log_q L\rfloor+1$.
 \end{lem}

 The link stated in Lemma \ref{2.5} establishes  a two-way bridge. Results for erasure list deciding can be derived directly  from the existing results on generalized Hamming weight, and thus the applications of generalized Hamming weight are inherited.  In the meanwhile, some new properties for generalized Hamming wight can be obtained as well if one can develop some fresh ideas on erasure list decoding.

In \cite{Gu}, Guruswami made use of the connection between generalized Hamming weight and erasure list decoding to establish some  bounds for rate $R_L(\tau)$ through the existing  bounds
on generalized Hamming weight.

  \begin{lem}\label{2.7}(see \cite{Gu}) One has
  \begin{itemize}
 \item [(i)]For every integer $L\geq 1$ and every $\tau$, $0\leq\tau\leq1$,
 \[R_L(\tau)\geq1-\frac{\tau}{r}\log_q\frac{q^r-1}{q-1}-\frac{H_q(\tau)}{r}\]
 where $r=\lfloor\log_q L\rfloor+1$. In particular, for any small $\epsilon>0$ and $\tau\in(0,1)$, there exists a  $(\tau,O(\exp(\frac1{\epsilon})))$-erasure list-decodable code of rate $1-\tau-\epsilon$.
 \item[(ii)]For small $\epsilon>0$ and  $\tau$ with $0<\tau<1$, a $(\tau,L)$-erasure list-decodable code of rate $1-\tau-\epsilon$ must satisfy $L\ge \Omega(\frac1{\epsilon})$.
  \end{itemize}
  \end{lem}

\section{Random List Decodable Erasure Codes}
Random $(1-\epsilon, O(1/\epsilon))$-erasure list-decodable codes of rate $R=\Omega(\epsilon/\log(1/\epsilon))$ was discussed in \cite{Gu} by using a characterization of generator matrices of erasure list-decodable codes. However, the rate is quite small and actually is dependent on $\epsilon$.
In this section, we are going to show that for any $0\leq R\leq 1$ ($R$ is independent of $\epsilon$), with probability $1-q^{-\Omega(n)}$  a random linear code $C$ of length $n$ and rate $R$ is $(1-R-\epsilon, 2^{O(1/\epsilon)})$-erasure list-decodable. Our approach is through a characterization of parity-check matrices of erasure list-decodable codes.

\begin{prop}\label{3.1} If $k/n\rightarrow R>0$ when $n$ tends to $\infty$, then for a random matrix $H$ over $\F_q$ of size ${(n-k)\times n}$, the probability that $H$ is full-rank is approaching $1$ when $n$ tends to $\infty$.\end{prop}
\begin{proof} On one hand,  it is easy to compute that the total number of  random matrix $H$ over $\F_q$ of size ${(n-k)\times n}$ with full rank is
$$(q^n-1)(q^n-q)\cdots(q^n-q^{n-k-1}).$$
On the other hand,  the total number of matrices $H$ over $\F_q$ of size ${(n-k)\times n}$ is $q^{n(n-k)}$. Let $E$ denote the event that an $(n-k)\times n$ random matrix $H$ over $\F_q$ is full-rank, then
$$Pr(E)=\frac{(q^n-1)(q^n-q)\cdots(q^n-q^{n-k-1})}{q^{n(n-k)}}.$$

To show $\lim_{n\rightarrow \infty}Pr(E)=1$, it suffices to show that $\lim_{n\rightarrow \infty}\ln Pr(E)\rightarrow 0.$

When $n$ tends to $\infty$, we have
\[0\ge \ln\frac{(q^n-1)(q^n-q)\cdots(q^n-q^{n-k-1})}{q^{n(n-k)}}=\sum_{i=k+1}^{n}\ln \left(1-\frac{1}{q^i}\right)
\geq\sum_{i=k+1}^n \left(-\frac{2}{q^i}\right)\geq-\frac{2n}{q^k}\rightarrow 0.\]
This completes the proof.
\end{proof}

\begin{lem}\label{3.2} Let $s$ be a positive integer, then an $[n,k]_q$ code $C$ is $(s/n,L)$-erasure-list-decodable if and only if any  submatrix $H^{'}_{(n-k)\times s}$ of the  parity check matrix $H_{(n-k)\times n}$ of $C$ has rank at least $s-\lfloor\log_q L\rfloor$.\end{lem}
\begin{proof}
 By Definition \ref{1.1} and the fact that $C$ is a linear code,  $C$ is $(s/n,L)$-erasure-list-decodable if and only if \[|\{\c\in C| \c_T=\bo\}|\leq L\]
 for $T\subseteq\{1,2,\dots,n\}$ with size $n-s$.
This implies that $C$ is $(s/n,L)$-erasure-list-decodable if and only if for any submatrix $H^{'}_{(n-k)\times s}$ of $H_{(n-k)\times n}$,
\[|\{\bx\in \F_q^s| H^{'}_{(n-k)\times s}\cdot \bx=0\}|\leq L,\]
 i.e., the solution space of $H^{'}_{(n-k)\times s}$ has dimension at most $\lfloor\log_q L\rfloor$. Therefore,  $H^{'}_{(n-k)\times s}$  has rank at least $s-\lfloor\log_q L\rfloor$.

\end{proof}

\begin{thm}\label{3.3} For every small $\epsilon>0$, a real $0<R< 1$ and sufficiently large $n$, with  probability at least $1-q^{-\Omega(n)}$,  a random linear code over $\F_q$  of length $n$ and rate $R$  is $(1-R-\epsilon,2^{O(\frac{1}{\epsilon})})$-erasure list-decodable.
\end{thm}
\begin{proof} Put $\ell=\left\lceil\frac{1}{\epsilon}((2-R)\log_q{2}+1)\right\rceil$ and $L=q^\ell$. Thus, $L=2^{O(\frac{1}{\epsilon})}$. We  randomly pick a  matrix $H_{(n-k)\times n}$.  Then with probability approaching $1$, $H_{(n-k)\times n}$ is full rank from Proposition \ref{3.1}. Let such a full rank matric $H_{(n-k)\times n}$ be  the parity check matrix of our linear code $C$. Then we are going to  prove that with probability at most $q^{-\Omega(n)}$, $C$ is not $(s,L)$-erasure list-decodable for $s=\lfloor n-k-\epsilon n\rfloor$. By Lemma \ref{3.2}, this happens only if some $(n-k)\times s$ submatrix of $H$ has rank less than  $s-\lfloor\log_q L\rfloor$.

Denote $n-k$ by $K$. Let $A$ denote the number of full-rank matrices $H_{(n-k)\times n}$ in which there exists $s=n-k-\epsilon n$ columns with rank at most $s-\ell$. Note that the total number of  matrices of size $K\times s$ over $\F_q$ with rank at most $s-\ell$
is equal to $\sum_{i=0}^{s-\ell}\binom{K}{i}(q^s-1)\cdots(q^s-q^{i-1})q^{(K-i)i}$.
Thus, we have
\begin{eqnarray*}A&\leq& \binom{n}{s}q^{K(n-s)}\sum_{i=0}^{s-\ell}\binom{K}{i}(q^s-1)\cdots(q^s-q^{i-1})q^{(K-i)i}\\
&<&2^n\sum_{i=0}^{s-\ell}\binom{K}{i}q^{(si+Ki-i^2)+K(n-s)}\\
&\leq&2^n\times q^{(s+K)(s-\ell)-(s-\ell)^2+K(n-s)}\sum_{i=0}^{s-\ell}\binom{K}{i}\\
&\leq&2^{n+K}\times q^{(s+K)(s-\ell)-(s-\ell)^2+K(n-s)}\\
&\leq&q^{(n+K)\log_q{2}}\times q^{(2K-\epsilon n)(K-\epsilon n-\ell)-(K-\epsilon n-\ell)^2+K(n-K+\epsilon n)}\\
&<&q^{n((2-R)\log 2-\epsilon \ell)+Kn}.\end{eqnarray*}

Substituting the value of $\ell$ to the above equation, we have $$\limsup_{n\rightarrow \infty}\frac{A}{(q^n-1)(q^n-q)\cdots(q^n-q^{n-k-1})}\le
\limsup_{n\rightarrow \infty}\frac{A}{q^{(n-k)n}}\times\lim_{n\rightarrow \infty}\frac{q^{(n-k)n}}{(q^n-1)(q^n-q)\cdots(q^n-q^{n-k-1})}\le \lim_{n\rightarrow \infty} q^{-n}=0.$$

This implies that with probability at most $q^{-n}$,  a random  matrix $H_{(n-k)\times n}$ has full rank  and  a submatrix of size $(n-k)\times s$ of rank at most $s-\ell$. The claimed result follows from  setting of our parameters.
\end{proof}

\section{Algebraic Geometry Codes are Good Erasure List-Decodable Codes}
In the previous section, we proved that random codes are good erasure list-decodable codes. There is still a lack of  constructive results on erasure list decoding. Though Guruswami \cite{Gu} presented a constructive result from concatenated codes, the rate is extremely small. In this section,  we  show that algebraic geometry (AG for short) codes are good erasure list-decodable codes  and furthermore they can be list decoded in  polynomial-time. As a preparation, we recall some basic results on AG codes first. Readers may refer to \cite{St93} for more details.

Let $\cX$ be a smooth, projective, absolutely irreducible curve of genus $g(\cX)$ (we will use $g$ instead of $g(\cX)$ if there is no confusion in the context) defined over $\F_q$. We denote by $\F_q(\cX)$ the function field of $\cX$. Denote by $N(\cX)$  the number of rational points of $\cX$. Let $\cP=\{P_1,...,P_n\}$ be a set of  $n$ distinct rational points over $\F_q$. Let $G$ be a divisor  such that ${\rm Supp}(G)\cap\{P_1,...,P_n\}=\emptyset$. Define $\cL(G)$ as the Riemann-Roch space associated to $G$ and denote $\dim \cL(G)=\ell(G)$. The algebraic geometry code $C(G,\cP)$ is defined as the image of $\cL(G)$ in $\F_q^n$ under the following evaluation map
$$C:\mathcal{L}(G)\longrightarrow \F_q^n,\quad f\mapsto (f(P_1),...,f(P_n)).$$
If $n>\deg G$, then $C(G,\cP)$ is an $[n,\ge\deg G-g+1,\ge n-\deg G]_q$-AG code. Throughout this section, we always assume that $n$ is bigger than $\deg(G)$.

The gonality of a curve   $\cX$  was introduced in \cite{Pe}. It is defined to be the  smallest degree of a nonconstant map  from $\cX$ to the projective line. We denote the gonality of $\cX$ by $t(\cX)$. More specifically, if $\cX$ is defined over a field $\F_q$ and $\F_q(\cX)$ is the function field of $\cX$, then $t(\cX)$ is the minimum degree of the field extensions of $\F_q(\cX)$ over a rational function field. It is easy to see that if $g(\cX)=0$, then  $t(\cX)=1$. If $g(\cX)=1$ or $ 2$, then  $t(\cX)=2$. However, for general $g$, the gonality is no longer determined by genus. In general, we have the following lower bound for $t(\cX)$.

\begin{lem}(\cite{Pe}) \label{4.1}
Let $\cX$ be a curve defined over $\F_q$ of genus $g$ with $N$ rational points. Then $ t(\cX)\geq N/(q+1)$.
\end{lem}

By using the lower bound on $t(\cX)$, one has the following proposition.
\begin{prop}\label{4.3}  $C(G,\cP)$ is $\left(\frac1n\left(n-\deg(G)+\lceil\frac{n}{q+1}\rceil-1\right),q\right)$-erasure list-decodable.\end{prop}
\begin{proof} Let $s$ be a positive integer with $s\leq n-\deg(G)+\lceil\frac{n}{q+1}\rceil-1$.  For any subset $T\subseteq\{1,2,\cdots,n\}$ of size $n-s$, we claim that $$|\{\c\in C(G,\cP)|\c_T=\mathbf{0}\}|\leq q.$$
This is equivalent to proving that $$\dim \mathcal{L}\left(G-\sum_{i\in  T}P_i\right)\leq 1.$$
Suppose $\dim \mathcal{L}\left(G-\sum_{i\in  T}P_i\right)\ge 2$, the one can choose a nonzero function $f\in \mathcal{L}\left(G-\sum_{i\in  T}P_i\right)$, then $$(f)+G-\sum_{i\in T}P_i\geq 0.$$

Let $H=(f)+G-\sum_{i\in  T}P_i\geq 0$. Then it is clear that $$\deg H=\deg\left(G-\sum_{i\in  T}P_i\right)=\deg(G)+s-n\le \left\lceil\frac{n}{q+1}\right\rceil-1\quad \mbox{and}\quad \dim\mathcal{L}(H)=\dim\mathcal{L}\left(G-\sum_{i\in  T}P_i\right)\ge 2.$$ Choose a function $z\in\mathcal{L}(H)\setminus\F_q$, then $[\F_q(\cX):\F_q(z)]$ is at most $\deg(H)\le \lceil\frac{n}{q+1}\rceil-1<\frac{N}{q+1}$. This contradicts Lemma \ref{4.1}.

Our desired result follows from Definition \ref{1.1}.
\end{proof}

Proposition \ref{4.3} can be extended by the Grismer bound through the following lemma.

\begin{lem}\label{4.5}
If  a divisor $G$ satisfies $\ell(G)\geq t\ge 1$ and $\deg G<N$, then $\deg G\geq N\cdot \frac{q^{t-1}-1}{q^t-1}$, where $N$ stands for the number of  rational points on $\cX$.
\end{lem}
\begin{proof} Suppose $P_1,\dots,P_N$ are $N$  distinct  rational points on $\cX$. By the strong approximation theorem,
there exists $x\in \F_q(\cX)$ such that ${\rm Supp} ((x)+G)\cap \{P_1,\dots,P_N\}=\emptyset$. Then $\ell((x)+G)=\ell(G)$ and $\deg((x)+G)=\deg(G)$.
Thus, we can obtain an algebraic geometry code $C((x)+G,\{P_1,\dots,P_N\})$ with parameters $[N,\ell(G),d\geq N-\deg G]_q$. By the Grismer bound \cite{LX04}, we have
\[N\geq\sum_{i=0}^{\ell(G)-1}\left \lceil\frac{d}{q^i}\right\rceil\geq\sum_{i=0}^{t-1}\left \lceil\frac{d}{q^i}\right\rceil\geq (N-\deg G)\sum_{i=0}^{t-1}\frac{1}{q^i}.\]
Thus, the desired result follows from the above inequality.
\end{proof}

\begin{thm}\label{4.7}  If $G$ satisfies $\ell(G)\geq t\ge 1$ and $\deg G<n$, then $C(G, \cP)$  is  $\left(\frac1n\left(n-\deg(G)+\lceil\frac{q^{t-1}-1}{q^t-1}n\rceil-1\right),q^{t-1}\right)$-erasure list-decodable.\end{thm}
\begin{proof}
Let $s$ be an integer satisfying $s \leq n-\deg G+\lceil\frac{q^{t-1}-1}{q^t-1}n\rceil-1$. For any $T\subseteq\{1,2,\cdots,n\}$ of size $n-s$, we  have
\[\deg\left(G-\sum_{i\in T}P_i\right)=\deg G-|T| =\deg G-n+s\le \left\lceil\frac{q^{t-1}-1}{q^t-1}n\right\rceil-1<N\cdot \frac{q^{t-1}-1}{q^t-1}.\]
By Lemma \ref{4.5}, we have
\[\ell\left(G-\sum_{i\in T}P_i\right)\le t-1.\]
Our desired result follows from Definition \ref{1.1}.
\end{proof}

\begin{rem}
When $t=1$, Theorem \ref{4.7} shows that $C(G, \cP)$  is  $\left(\frac1n\left(n-\deg(G)-1\right),1\right)$-erasure list-decodable. For $t=2$, we obtain the result of Proposition \ref{4.3}.
\end{rem}

Combing  Lemma \ref{2.5} and Theorem \ref{4.7}, we immediately obtain the following lower bound on generalized Hamming weight of algebraic geometry codes.
\begin{cor} For $1\leq t\leq \deg(G)-g+1$, the $t$-th generalized Hamming weight of $C(G,\mathcal{P})$ satisfies $$d_t(C(G,\mathcal{P}))\ge n-\deg(G)+\left\lceil\frac{q^{t-1}-1}{q^t-1}n\right\rceil.$$ \end{cor}

Now we come to the main result of this section.
\begin{thm}\label{4.8}
Let $q$ be a square. For any small $\epsilon>0$ and $\tau$ with $0<\tau<1-\frac{1}{\sqrt{q}-1}+\frac{1}{q}-\epsilon$, there exists  a family $\{C(G,\cP)\}$ of algebraic geometry code with length tending to $\infty$ such that $C(G,\cP)$ have rate at least $1-\tau-\frac{1}{\sqrt{q}-1}+\frac{1}{q}-\epsilon$ and are  $(\tau, O(\frac1{\epsilon}))$-erasure list-decodable. Furthermore, it can be list decoded in $O((n\log_qn)^3)$ time, where $n$ is the length of the code.
\end{thm}
\begin{proof} Choose a curve $\cX/\F_q$ in the Garcia-Stichtenoth tower \cite{Ga}. Then $N(\cX)/g(\cX)\rightarrow\sqrt{q}-1$. Let $\cP=\{P_1,P_2,\dots,P_n\}$ with $n=N(\cX)-1$. Choose the last rational point $P$ of $\cX$ such that $P\not\in\cP$.  Put
 \[m:=n-\lceil\tau n\rceil+\left\lceil\frac{q^{t-1}-1}{q^t-1}n\right\rceil-1\]
 and $G=mP$. By Theorem \ref{4.7},  $C(G, \cP)$  is  $\left(\frac1n\left(n-m+\lceil\frac{q^{t-1}-1}{q^t-1}n\rceil-1\right),q^{t-1}\right)$-erasure list-decodable for any constant $t\ge 1$. Hence,  $C(G, \cP)$  is $(\tau, q^{t-1})$-erasure list-decodable. Pick $\epsilon=\frac{1}{q}-\frac{q^{t-1}-1}{q^t-1}=\frac{q-1}{q(q^t-1)}$, then $q^{t-1}=O(\frac1{\epsilon})$. Moreover, the rate of $C(G, \cP)$ is at least
 \[\frac{1}{n}(m-g+1)\rightarrow 1-\tau-\frac{1}{\sqrt{q}-1}+\frac{q^{t-1}-1}{q^t-1}=1-\tau-\frac{1}{\sqrt{q}-1}+\frac{1}{q}- \epsilon.\]
   This proves the first statement of the theorem.

   Finally by \cite{Shu}, we know that a basis of $\cL(G)$ can be found in  $O((n\log_qn)^3)$ time, where $n$ is the length of the code. Assume that we have already found a basis $f_1,\dots,f_k$ of $\cL(G)$. Suppose that $\bc=(c_1,\dots,c_n)$ was transmitted and $\bc_T$ was received with $T\subseteq\{1,2,\dots,n\}$ and $|T|\ge (1-\tau)n$. A function $f\in \cL(G)$ is in the list if and only if  $f(P_i)=c_i$ for all $i\in T$. Let $f=\sum_{j=1}^k\lambda_j f_j$ with $\lambda_j$ being unknowns. Then one has $\sum_{j=1}^k\lambda_j f_j(P_i)=c_i$ for all $i\in T$. This is a system of linear equations with $|T|$ equations and $k$ unknowns. It can be solved in $O(n^3)$ time. This completes the proof.
\end{proof}

\end{document}